\documentclass[12pt]{article}
\usepackage{amsmath, amssymb, bbm}
\usepackage{graphicx,psfrag,epsf}
\usepackage{enumerate}
\usepackage[ sort, authoryear]{natbib}
\usepackage{url} 

\newcommand{\blind}{0}

\newtheorem{theorem}{Theorem}[section]

\newtheorem{remark}{Remark}
\newtheorem{proof}{Proof}
\newtheorem{definition}{Definition}

\begin{document}

\def\spacingset#1{\renewcommand{\baselinestretch}%
{#1}\small\normalsize} \spacingset{1}


\if0\blind
{
  \title{\bf One-dimensional Nonstationary Process Variance Function Estimation}
  \author{Eunice J. Kim 
  \hspace{.2cm}\\
    Department of Mathematics and Statistics, Amherst College\\
    and \\
    Zhengyuan Zhu \\
    Department of Statistics, Iowa State University}
  \maketitle
} \fi

\if1\blind
{
  \bigskip
  \bigskip
  \bigskip
  \begin{center}
    {\LARGE\bf Title}
\end{center}
  \medskip
} \fi

\bigskip
\begin{abstract}
Many spatial processes exhibit nonstationary features. We estimate a variance function from a single process observation where the errors are nonstationary and correlated. We propose a difference-based approach for a one-dimensional nonstationary process and develop a bandwidth selection method for smoothing, taking into account the correlation in the errors. The estimation results are compared to that of a local-likelihood approach proposed by Anderes and Stein(2011).  A simulation study shows that our method has a smaller integrated MSE, easily fixes the boundary bias problem, and requires far less computing time than the likelihood-based method.
\end{abstract}

\noindent%

\newpage
\spacingset{1.45} 
\section{Introduction}
\label{sec:introduction}

The prevalence of mobile devices and increase in storage capacity have brought about high demand for spatial data analysis. 
Many spatial processes exhibit nonstationary features, such as non-constant mean and variance and changing structure of autocorrelation. We encounter these features in data from ecology, geology, meteorology, astronomy, and economics to name a few. Specific examples include processes describing natural phenomena, such as species and mineral dispersal, wind fields, crop yields, and the cosmic microwave background, 
as well as human activities in aggregate, such as geolocated Internet search queries, real estate prices, and air pollution. 
When analyzing these data, it is not only important to estimate the overall trend in the process but also useful to construct reasonable interval estimates of the mean process and provide spatial prediction intervals. 

In this paper, we are interested in estimating the variance function of a one-dimensional spatial process where the mean and the variance functions are smooth and have additive correlated errors. We assume a fixed equidistant grid design for a one-dimensional process and consider a mixed domain asymptotic framework to develop the theory for a variance function estimator. \cite{Brown:2007ij} have discussed asymptotic properties of nonparametric variance estimators formed by differencing. This article extends the scenario to nonstationary correlated error processes and discusses cross-validation for bandwidth selection. Our estimator requires estimating the correlation structure embedded in the data and adjusting the scale of the difference-based estimator using the estimated correlation. We describe the asymptotic properties of our estimator and evaluate its performance using a simulation study. 

Section~\ref{sec:related} discusses prior work on variance estimation for nonstationary processes. 
Section~\ref{sec:model} defines a data model and local variogram as a product of a variance function and a standard variogram function. 
Section~\ref{sec:theoretical}  looks into the estimator of local variogram function and its theoretical properties, and Section~\ref{sec:algorithm} presents the algorithm for variance function estimation. Section~\ref{sec:simulation} evaluates the method through a simulation study, and Section~\ref{sec:conclusion} discusses the advantages of the difference-based variance function estimator comparing it to a likelihood-based estimator and closes with future work. 

\section{Related Work}
\label{sec:related}

\cite{Neumann:1941fk} proposed using differences of successive observations to estimate the variance of independent and identically distributed ($i.i.d.$) errors. \cite{Seifert:1993oq} and \cite{Wang:2008} explored that the bias in the variance estimation was reduced via differencing as it cancels out a mean or a gradually changing mean function.   \cite{Gasser:1986ly} proposed second-order differencing to estimate a variance function incorporating irregularly-spaced observations, and \cite{Hall:1990zr, Hall:1991ai} used a differencing approach to estimate the variance of two-dimensional processes with $i.i.d.$ errors in image processing. All these work assume independence among the errors.

We propose a difference-based variance function estimator for one-dimensional processes where the errors are nonstationary and correlated. Under the same assumption of data model, \cite{Anderes:2011nx} proposed a likelihood-based method to estimate a smooth variance function and a correlation function. Their method can handle irregularly spaced data and provides statistical efficiency when the Gaussianity assumption is met for the observed process.  Still, the computational burden is heavy especially when selecting bandwidth as covariance matrices must be inverted for every iteration of simulation. Gaussian distributional assumption could also be stringent for the observed nonstationary processes.

\cite{Hall:1989ys} discussed the asymptotic risk of the difference-based variance function estimator in nonparametric regression depending on the smoothness of a variance function relative to the smoothness of a mean function. \cite{Brown:2007ij} re-examined the asymptotic properties of difference-based variance estimators of non-constant and independent errors, and
 \cite{Wang:2008} derived the asymptotic minimax risk rate in terms of the the degree differentiability of the mean and variance function, $\alpha$ and  $\beta$ respectively.  When $\alpha$ is greater than or equal to $1/4$, the convergence rate of risk is $O(n^{-\beta/(2\beta+1)})$, the same as in nonparametric regression with $i.i.d.$ errors; and when $\alpha$ is less than $1/4$, then the risk is $O(n^{-4\alpha})$, still slightly larger than $O(n^{-1})$. In Section ~\ref{sec:theoretical} we describe the asymptotic results of the correlated error scenario and compare it to the independent error scenario.

 In signal processing, a band-pass filter provides local variance estimation assuming that the marginal mean function is changing slowly. The method works well under second-order stationarity of the error process but not under nonstationarity. Since the shape of a filter and the passband interact with the underlying data process, the estimation result gets distorted systematically. Also, converting the output from frequency domain may introduce bias. Therefore, working in the time-domain for the output in the same domain should result in more accurate estimation. 

\section{Data Model and Definition}
\label{sec:model}
We define a continuously differentiable Lipschitz function, a nonstationary data model, and local variogram. Estimating a variance function using a nonparametric approach, the Lipschitz condition on the mean and variance functions of the data provides the basis of minimum smoothness.
\begin{definition}
\label{def:lip}
Let $c_{1},c_{2}>0$. Denote $q'\doteq q-\lfloor q\rfloor$  where $\lfloor q\rfloor$  is the largest integer less than $q$. We say that the function $f(x)$ is in class of $\Lambda_{q}(c_f)$  if for all $x,y\in(0,1)$ , $\left|f^{(\lfloor q\rfloor)}(x)-f^{(\lfloor q\rfloor)}(y)\right|\leq c_{1}\left|x-y\right|^{q'}$, $\left|f^{(k)}(x)\right|\leq c_{2}$  for $k=0,\dots,\lfloor q\rfloor$, and $c_f=\max(c_1, c_2)$. 
\end{definition}
\begin{definition}
\label{def:lipplus}
If a function $f(x)$  is in class $\Lambda_{q}(c_f)$  and there exists $\delta>0$  such that $f(x)>\delta$  for all $x\in[0,1]$, we say the function is in $\Lambda_{q}^{+}(c_f)$.
\end{definition}

Consider a nonstationary continuous process model 
\begin{equation}
\label{eq:data-model} Z(s) = \mu(s)+\sigma(s)X(s) 
\end{equation} 
on $0\leq s \leq1$, without loss of generality,  with a smooth mean function $\mu(s)$ and an additive, correlated noise as a product of a smooth standard deviation function $\sigma(s)$ and a second-order stationary process $\{X(s)\}$  where $E(X(s))=0,$ $var(X(s))=1$, and $cov(X(s),X(s'))=\rho(\left|s-s'\right|; \theta)$ for all pairs of $s$ and $s'$ in the unit interval. We consider $\mu(s)\in \Lambda_{q}(c_f)$,  $q \geq 0$, and $\sigma^{2}(s)\in \Lambda_{\beta}^{+}(c_f)$ , $\beta \geq 2$.  We assume the following general form of a correlation function:
\begin{align}
\rho(\left|s-s'\right|;\theta)=\begin{cases}
1 & s = s' \\
1- \dfrac{ \left|s-s'\right|^{\alpha}}{\theta}+O(\left|s-s'\right|^{\alpha+2}) & s \neq s'
\end{cases} 
\label{eq:rho}
\end{align}
where $\theta>0$  and $0<\alpha<2$ for a valid correlation structure. This class (\ref{eq:rho}) of correlation function encompasses linear, spherical, Mat\'ern and exponential models. We assume an equally spaced design and define $s_{i} =  (2i-1)/(2n)$  where the location is indexed by $i=1,\dots,n$. As a shorthand, we write $Z_{i} =  Z\left(s_{i}\right),$ $ \mu_{i}=\mu\left(s_{i}\right),$ $\sigma_{i}=\sigma\left(s_{i}\right), \rho_{h} = \rho\left(h/n \right)$, and to specify a parametric correlation function, the shorthand is $\rho_{s; \theta}= \rho\left(s;\theta\right)$. We use the following notation to denote the $j$\textsuperscript{th}-order derivative of a function $\sigma^2(s)$: $\sigma^{2(j)}(s)=\left. d^{j} \sigma^{2}(x) / dx^{j} \right|_{x=s} $. 

We expand the definition of \textit{variogram}, introduced by \cite{Matheron:1962}, since differencing nonstationary process data requires local treatment.  Using a 0-mean nonstationary process as a data model from (\ref{eq:data-model}), the variance at $s$ of a lag-$h$ first-order differenced process is 
\begin{align}
\label{eq:deriving-locvariog}
 var&\left(Z\left(s-\frac{h}{2n}\right)-Z\left(s+\frac{h}{2n}\right)\right) \nonumber \\
= & 2\sigma^{2}(s) \left(1- \rho_{h} \right) + 2 \left(\sigma^{(1)}(s)\right)^2  \left(1+ \rho_{h} \right) \left(\frac{h}{2n}\right)^2 +  o\left(n^{-2}\right).  
\end{align}
The first term resembles the definition of a variogram but with local variance, and the derivatives of the smoothly changing local variance and other higher order terms follow.
\begin{definition}
The local variogram $2\gamma_{L}\left(s,h;\theta\right)$  is defined as the leading term of (\ref{eq:deriving-locvariog}), i.e. 
\begin{equation}
\label{eq:Def-locvariog}
\gamma_{L}(s,h;\theta) = \sigma^{2}(s)\left(1-\rho\left(\frac{h}{n};\theta\right)\right).
 \end{equation} 
\end{definition}
Local variogram (\ref{eq:Def-locvariog}) is a product of a variance function and the variogram of a standardized process. Variogram represents spatial dispersion by taking lagged differences of a stationary process, and local variogram describes spatial dispersion of a nonstationary process. When the lag size $h$ is  small in comparison to the number of observed points $n$ in a process, that is in mixed-domain asymptotic,  the higher order terms vanish. 

%

\section{Theoretical Results}
\label{sec:theoretical}

We first define an estimator for local variogram in Section ~\ref{sec:lvariog} as a preliminary steps to estimating a variance function. The bias and the variance of the estimator are derived in Section ~\ref{sec:Bias} and ~\ref{sec:Var} respectively. In Section 4.4 the asymptotic convergence rate of the point-wise mean square error of the local variogram estimator is compared to that of a standard nonparametric estimator with i.i.d. errors.


\subsection{Local variogram estimator}
\label{sec:lvariog}
If $\{Z_i\}$ is an independent and identically distributed error process with mean 0 and variance 1, this implies that taking a simple differencing by lag-$h$ of an equally-spaced process renders a correlated sequence $\{D_{i,h}\}_{i=1}^{n-h} =\{\left(Z_{i}-Z_{i+h}\right)/\sqrt{2}\} _{i=1}^{n-h}$ where $E(D_{i,h})=0$ and $var(D_{i,h})=1$ for $i=1,\dots, n-h$ and any  positive integer $h < n$. \cite{Brown:2007ij} refer to the sequence $\{D_{i,h}\}_{i=1}^{n-h}$ as pseudo-residuals. The squared pseudo-residuals  at different lags resemble a variogram since $var(D_{i,h}) = E(D_{i,h}^2)$ by the 0-mean property. We derive an estimator of local variogram from the differenced process since the $\{Z_i\}$ we consider is nonstationary as in (\ref{eq:data-model}).

Let  $K_{\lambda}$ represent a Gasser-M\"uller kernel with bandwidth $\lambda$ and restrict its order to be greater than $\beta$, the degree differentiability of variance function in $\{Z_i\}$. \cite{Gasser:1985} have developed the kernel so that the moment conditions simplify the calculation of high-order terms in nonparametric estimators and that the edge effect be easily removed by adjusting the kernels at the boundaries of the domain.
 Without loss of generality, we assume that the process is observed on a unit interval, $[0,1]$. We normalize the simple differenced process by scaling it by $1/\sqrt{2}$ so that $var(D_{i,h})$ is in the same order as $var(Z_i)$, and we match the $i$\textsuperscript{th} pseudo-residual at the center of observed pair location, that is $(s_i+s_{i+h})/2 = (i + h/2)/n$.

We define Gasser-M\"uller kernel estimator of local variogram at location $s$ and lag $h$ as 
\begin{equation}
\label{eq:variog-est} 
\hat{\gamma}_{L\:\lambda}(s,h)=\sum_{i=1}^{n-h}K_{\lambda,i+h/2}(s)D_{i,h}^{2}.
\end{equation} 
%
 Note that in the local variogram estimator (\ref{eq:variog-est}) the $i^{th}$ squared difference $D^2_{i,h}$ is associated with the kernel weight indexed by $i+h/2$. The shift in the index by $h/2$ is to align the portion of kernel weight to the center of the observed pair in constructing $D^2_{i,h}$. When $h=1$, for example, the limits of the kernel are at $s_{i}$  and $s_{i+1}$ corresponding to the locations of observations  $Z(s_{i})$ and $Z(s_{i+1})$, which form the $i$\textsuperscript{th} pseudo-residual. If the kernel weights were indexed by $i$ instead of $i+h/2$, then the lower and upper integration limits will be  $(s_{i-1}+s_{i})/2$ and $(s_{i}+s_{i+1})/2$ respectively, whose locations are translated by -$h/2$.

It is possible to consider higher-order differencing to form pseudo-residuals, but the first-order differencing introduces the least bias and variance in local variogram estimation with the least number of correlated terms. We also suggest using the smallest lag in differencing because it reduces correlation in the newly constructed sequence of pseudo-residuals and helps with the estimation of embedded correlation.


\subsection{Bias of the estimator}
\label{sec:Bias}
Let $D_{i,h}=\left(Z_{i}-Z_{i+h}\right)/\sqrt{2}$, $\delta_{i,h}  =  \mu_{i}-\mu_{i+h}$, and $g_{i,h}= \sigma_{i}^{2}+\sigma_{i+h}^{2}-2\sigma_{i}\sigma_{i+h}\rho_{h}$ for $i=1,\dots,n-h$. The expected value of the local variogram estimator is 
\begin{eqnarray*}
 E\left(\hat{\gamma}_{L \, \lambda}(s,h)\right) & = &\sum_{i=1}^{n-h}K_{\lambda,i+\frac{h}{2}}(s)E\left(D_{i,h}^{2}\right)\\
  & = & \frac{1}{2}\sum_{i=1}^{n-h}K_{\lambda,i+\frac{h}{2}}(s)\left\{(\mu_{i}-\mu_{i+h})^{2} + \sigma_{i}^{2}+\sigma_{i+h}^{2}-2\sigma_{i}\sigma_{i+h}\rho_{h}\right\}.
 \end{eqnarray*} 
The bias of the local variogram estimator is
\begin{align}
bias(\hat{\gamma}_{\lambda}(s,h)) & =E(\hat{\gamma}_{\lambda}(s,h))-(1-\rho_{h})\sigma^{2}(s)\nonumber \\
& = \sum_{i=1}^{n-h}K_{\lambda,i+\frac{h}{2}}(s)\left\{\frac{1}{2}(\delta_{i,h}^{2}+g_{i,h})-(1-\rho_{h})\sigma^{2}(s) \right\}.
\label{eq:bias0}
\end{align}
Note that $(1-\rho_{h})=O(n^{-\alpha})$ and $0<\alpha < 2$. 
\begin{theorem}
\label{theorem1}
Assume a nonstationary data model (\ref{eq:data-model}) and the correlation function (\ref{eq:rho}). The mean and the variance functions $\mu(s)$ and $\sigma^{2}(s)$ are continuously differentiable Lipschitz functions (see Definitions \ref{def:lip} and \ref{def:lipplus}) where $\mu(s)\in \Lambda_{q}(c_f), q \geq 0$  and $\sigma^{2}(s)\in \Lambda_{\beta}^{+}(c_f), \beta \geq 2$. The difference-based local variogram $m$-order Gasser-M\"uller kernel estimator (\ref{eq:variog-est}) at location $s$ and lag $h$ has an asymptotic bias of order 
\begin{equation}
\label{eq:bias1}
bias(\hat{\gamma}_{\lambda}(s,h)) = 
\begin{cases} O(n^{-2} +  n^{-2q} + n^{-\alpha-1})  &  \mbox{ where } \quad q, \beta < m\\
O(n^{-2} +  n^{-2q} + n^{-\alpha-1})+ O(n^{-\alpha}\lambda^{m})  &  \mbox{ where } \quad q < m \leq \beta\\
O(n^{-2} + n^{-2q} + n^{-\alpha-1}) + O(\lambda^{ m}) & \mbox{ where } \quad m \leq q.
\end{cases}
\end{equation}
\end{theorem}
\begin{proof}
To calculate an asymptotic bias we split (\ref{eq:bias0}) into two parts. 
The first term is $\delta_{i,h}^{2}$ whose expansion is in (\ref{eq:del-expan}) for $q\geq 1$ and in (\ref{eq:del-expanq}) for $0 \leq q < 1$.  Convolved with a Gasser-M\"uller kernel of order $m$ (\cite{Gasser:1985}), the higher order terms in $\delta_{i,h}^2$ cancel when the number of derivatives of the mean function $q \leq m$:
\begin{align}
\sum_{i=1}^{n-h} K_{\lambda,i+\frac{h}{2}}(s)\delta_{i,h}^{2} 
&= \begin{cases} O(n^{-2}) + O \left(n^{-2q}\right)  &  \mbox{ where } \quad q < m\\
O(n^{-2}) + O \left(n^{-2q}\right)+O(\lambda^{m}) & \mbox{ where } \quad q \geq m.
\end{cases}
\label{eq: b-part1}
\end{align}
The second part of the bias is $\frac{1}{2}g_{i,h} - \sigma^{2}(s)(1-\rho_{h})$. In equation (\ref{eq:g-expan}), the leading term in the expansion of $g_{i,h}$ about $s$  is the local variogram $\sigma^{2}(s)(1-\rho_{h})$. Applying a Gasser-M\"uller kernel to the remaining high order terms in (\ref{eq:g-expan}), we get the following:
\begin{align}
&\sum_{i=1}^{n-h}K_{\lambda,i+\frac{h}{2}}(s) \left\{ \frac{1}{2}g_{i,h}  - \sigma^{2}(s)(1-\rho_{h})  \right\} \nonumber \\
& = \sum_{i=1}^{n-h}K_{\lambda,i+\frac{h}{2}}(s) \left\{ (1-\rho_{h})\left(\sigma_{s}\sigma_{s}^{(1)}\frac{h}{n} +  \frac{\sigma_{s}\sigma_{s}^{(2)} }{2}\frac {h^2}{n^2} \right) +  \frac{1}{2}\left(\sigma_{s}^{(1)} \frac {h}{n}  \right)^{2} \right\} \nonumber \\
& \quad + \sum_{i=1}^{n-h}K_{\lambda,i+\frac{h}{2}}(s) (1-\rho_{h})\sum_{j=1}^{\lfloor\beta\rfloor}  \left\{ \frac{ \left(\sigma^{2}_{s}\right)^{(j)}}{j!} +\frac{ \left(\sigma^{2}_{s}\right)^{(j+1)}}{2(j+1)!} \left(1 + \frac{h}{n} \right)\frac {h}{n} \right\} (s_{i}-s)^{j}  \nonumber \\
& \quad  + \sum_{i=1}^{n-h}K_{\lambda,i+\frac{h}{2}}(s) \frac{h^2}{2n^2}  \sum_{k=1}^{\lfloor\beta\rfloor} \sum_{j=1}^{k+1} c_{k}\sigma^{(j)}_{s}\sigma^{(k-j+2)}_{s} (s_{i}-s)^{k}  + \sum_{i=1}^{n-h}K_{\lambda,i+\frac{h}{2}}(s)O(|s_{i}-s|^{\beta}) \nonumber \\
& \quad = \begin{cases} 
O\left(n^{-\alpha-1}\right) + O\left(n^{-2}\right)  &  \mbox{ where } \beta < m  \\
O\left(n^{-\alpha-1}\right) + O\left(n^{-2}\right) +O(n^{-\alpha}\lambda^m)   &   \mbox{ where } \beta \geq m.
\end{cases}
\label{eq: b-part2}
\end{align}
The bias in the local variogram is summarized in (\ref {eq:bias1}) combining the results in (\ref{eq: b-part1}) and (\ref{eq: b-part2}).
\end{proof}
As we see from the asymptotic results summarized in (\ref{eq:bias1}), the bias has an order dependent on the differentiability of the mean and the variance functions $q$ and $\beta$ respectively, the order $m$ of the kernel, and the smoothness of the process defined by $\alpha$. As $\alpha$ approaches 0, the data process is rough and work almost as an independent process. As $\alpha$ approaches 2, the smoothness increases and follows the form of a process with a Gaussian correlation function. Also note that when $\alpha=1$, the smoothness is equivalent to a process generated with an exponential correlation function, which we use in the simulation study.

\begin{remark}
We detail Theorem \ref{theorem1} in the order we listed the results in in (\ref{eq:bias1}).
\begin{itemize}
\item [A.] Assume that $m > q$ and $m > \beta$, in other words the order of kernel is greater than the degree differentiability of both the mean and variance functions.  Then, (A.i) when $\alpha < 1$ and $\frac{\alpha + 1}{2} < q \leq 1$,  the bias is $O\left(n^{-\alpha-1}\right)$; (A.ii) when $\alpha < 1$ and $2q \leq \alpha + 1/2$,  the bias is $O\left(n^{-2q}\right)$; and (A.iii) when $\alpha \geq 1$ and $q \geq 1$,  the bias is $O\left(n^{-2}\right)$.

\item [B.] Assume that $q < m \leq \beta $ and that $\lambda=O(n^{-x})$ where $0<x<1$. Then $O(n^{-\alpha}\lambda^{m})$ is the order of bias in the following three settings: (B.i) $ q \geq 1$, $\alpha \leq 1$, and $x < 1/m$; (B.ii) $q \geq 1$, $\alpha \geq 1$, and $x < (2-\alpha)/m$; and (B.iii) $\alpha < 1$, $2q < \alpha+1$, and $x < (2q - \alpha)/m $. The remaining scenarios should should refer to Case A.

\item [C.] Assume that $m \leq q$ irrespective of $\beta$ and that $\lambda=O(n^{-x})$ where $0<x<1$.  Then the bias is $O(\lambda^{m})$ in the following three settings: (C.i) $q\geq1$, $\alpha \geq 1$, and $2/m > x$; (C.ii) $q < \min(1, \frac{\alpha +1}{2})$, and $x < 2q/m$; (C.iii) $\alpha <1$, $\alpha+1<2q$ and $x < (\alpha + 1)/m$. The remaining scenarios should refer to Case A. 
\end{itemize}
\end{remark}

When $m$ is greater than $q$ and  $\beta$, which is the case of A, the asymptotic bias is the smallest. Therefore, we recommend choosing a high order kernel function even though we do not know $q$ and $\beta$ in practice. Following the recommendation, it will be very unlikely to encounter the bias in $O(\lambda^{m})$ of case C.

\subsection{Variance of the estimator}
\label{sec:Var}
The variance of the local variogram estimator at location $s$  and lag $h$  is 
\begin{equation}
\label{eq:var-raw}
var(\hat{\gamma}_{\lambda}(s,h))=\sum_{i=1}^{n-h}\sum_{j=1}^{n-h}K_{\lambda,i+\frac{h}{2}}(s)K_{\lambda,j+\frac{h}{2}}(s) cov(D_{i,h}^{2},D_{j,h}^{2}). 
\end{equation}
Recall $D_{i,h}=\left(\delta_{i}+\sigma_{i}X_{i}-\sigma_{i+h}X_{i+h}\right)/\sqrt{2}$ where $X_{i}$ is a stationary process with mean 0, variance 1, and a correlation function $cov(X_{i},X_{i+h}) = \rho_{h}$. Let $\{X_i\}_{i=1}^{n}$  be a Gaussian process. Then $\left(\sigma_{i}X_{i}-\sigma_{i+h}X_{i+h}\right)$ is distributed $Normal\left(0,g_{i,h}\right)$ and its fourth moment is $E\left(\sigma_{i}X_{i}-\sigma_{i+h}X_{i+h}\right)^{4}=3g_{i,h}^{2}$. The variance of the squared pseudo-residual is 
\begin{eqnarray*}
var(D_{i,h}^{2}) 
&=&E(D_{i,h}^{4})-E^2(D_{i,h}^{2})\\
&=&\frac{1}{4}\left\{ \delta_{i,h}^{4}+6\delta_{i,h}^{2}g_{i,h}+3g_{i,h}^{2}-\left(\delta_{i,h}^{2}+g_{i,h}\right)^{2}\right\} \\
&=&\delta_{i,h}^{2}g_{i,h}+\frac{1}{2}g_{i,h}^{2}.
\end{eqnarray*}
The covariance between the $i$\textsuperscript{th} and the $j$\textsuperscript{th} squared differences is
{\allowdisplaybreaks 
\begin{align*}
cov&(D_{i,h}^{2},D_{j,h}^{2})\\
=&\frac{1}{4}\left\{E(\left(Z_{i}-Z_{i+h}\right)^2\left(Z_{j}-Z_{j+h}\right)^2)-(\delta_{i,h}^{2}+g_{i,h})(\delta_{j,h}^{2}+g_{j,h})\right\}\\
=& \delta_{i,h}\delta_{j,h}\{\rho_{|i-j|}(\sigma_{i}\sigma_{j}+\sigma_{i+h}\sigma_{j+h}) - \rho_{|i-j-h|}\sigma_{i}\sigma_{j+h} - \rho_{|i-j+h|}\sigma_{i+h}\sigma_{j}\}\\
 &+  \frac{1}{2}\{(\rho_{|i-j|}\sigma_{i}\sigma_{j} - \rho_{|i-j-h|}\sigma_{i}\sigma_{j+h})^{2} + 
\left( \rho_{|i-j+h|}\sigma_{i+h}\sigma_{j} - \rho_{|i-j|}\sigma_{i+h}\sigma_{j+h}\right)^{2}\}\\
& + (\rho_{|i-j|}^{2}+\rho_{|i-j-h|}\rho_{|i-j+h|})\sigma_{i}\sigma_{i+h}\sigma_{j}\sigma_{j+h}
-\rho_{|i-j|}\sigma_{i}\sigma_{i+h}(\rho_{|i-j+h|}\sigma_{j}^{2} + \rho_{|i-j-h|}\sigma_{j+h}^{2})\\
= &\delta_{i,h}\delta_{j,h}P_{ij}+\frac{1}{2}P_{ij}^{2}.
\end{align*}}
\noindent where $P_{ij}= \rho_{|i-j|}\left(\sigma_{i}\sigma_{j}+\sigma_{i+h}\sigma_{j+h}\right) - \rho_{|i-j-h|}\sigma_{i}\sigma_{j+h} - \rho_{|i-j+h|}\sigma_{i+h}\sigma_{j}$ for $i \neq j$. Note that when $i = j$, $P_{ii} = g_{i,h}$. 
The Taylor expansion of $P_{i,j}$ about $s_i$ for any $i\neq j$ is
\begin{equation}
\label{eq:Pij}
P_{ij} = \frac{h^2}{n^2} \left( \sigma_i^{(1)}\right)^2 - \frac{2h^2}{(n\theta)^2} \sigma_i^2 + o\left(n^{-3} \right).
\end{equation}

In Theorem \ref{theorem2} we derive the asymptotic rate of convergence of the variance of Gasser-M\"uller estimator of local variogram expressed in (\ref{eq:var-raw}).
\begin{theorem} 
\label{theorem2}
Assume the same conditions as in Theorem \ref{theorem1}. The variance of the local variogram estimator $\hat{\gamma}_{L,\lambda}(s, h)$ in (\ref{eq:variog-est}) is asymptotically in the order as follows:
{\allowdisplaybreaks 
\begin{align}
var( \hat{\gamma}_{\lambda}(s,h)) 
& =O\left(\frac{1}{n\lambda}\right) O \left( n^{-2q-\alpha} + n^{-2\alpha} \right).
\label{eq:var1} 
\end{align}}
\end{theorem}
\begin{proof} 
Use Taylor expansions about $s$ of $\delta_{i,h}$, $g_{i,h}$,  and $P_{ij}$  (details are in equations (\ref{eq:del-expan}) and (\ref{eq:g-expan}) in the Appendix and in (\ref{eq:Pij}) respectively), and obtain a Taylor expansion of the variance in (\ref{eq:var-raw}) about $s$ at fixed lag $h$. 
{\allowdisplaybreaks 
\begin{align}
\label{eq:var0}
var(\hat{\gamma}_{\lambda}(s,h)) 
=&\sum_{i=1}^{n-h}K_{\lambda,i+\frac{h}{2}}^{2}(s)\left(\delta_{i}^{2}g_{i}+\frac{g_{i}^{2}}{2}\right) +2\sum_{i> j=1 }^{n-h-1}K_{\lambda,i+\frac{h}{2}}(s)K_{\lambda,j+\frac{h}{2}}(s)\left(\delta_{i}\delta_{j}P_{ij}+\frac{P_{ij}^{2}}{2}\right) \nonumber \\
 = & 2 \sum_{i=1}^{n-h}K_{\lambda,i+\frac{h}{2}}^{2}(s) \left\{ \delta_i^{2} (1-\rho_h)O\left(1\right) 
  +   (1-\rho_h)^2O\left( 1\right)  \right\} \nonumber \\
  & + 2 \sum_{i>j=1}^{n-h-1}K_{\lambda,i+\frac{h}{2}}^{2}(s)K_{\lambda,j+\frac{h}{2}}^{2}(s) \left\{ \delta_i \delta_j O(n^{-2}) + O(n^{-4}) \right\} \nonumber \\ 
 = & 2(1-\rho_h)  \sum_{i=1}^{n-h}K_{\lambda,i+\frac{h}{2}}^{2}(s)\left\{O( n^{-2} + n^{-2q})  + (1-\rho_h)O(1) \right\} \nonumber \\ 
& + 2\sum_{i>j=1}^{n-h-1}K_{\lambda,i+\frac{h}{2}}(s)K_{\lambda,j+\frac{h}{2}}(s) O(n^{-4}) 
\end{align}}
Use the fact that  $K_{\lambda,i+\frac{h}{2}} = O\left(\frac{1}{n\lambda}\right)$ 
and $\sum K_{\lambda,i+\frac{h}{2}}^{2} = O\left(\frac{1}{n\lambda}\right)$ and reduce the last line (\ref{eq:var0}) to (\ref{eq:var1}).
\end{proof}


The correlation between $D_{i,h}^2$ and $D_{j,h}^2$ is 
\begin{align*}
cor&(D_{i,h}^{2}, D_{j,h}^{2})\\
 =&\frac{cov(D_{i,h}^{2}, D_{j,h}^{2})}{\sqrt{var(D_{i,h}^{2}) var(D_{j,h}^{2})} } \\
 =& \frac{\delta_{i,h} \delta_{j,h} P_{ij} + \frac{1}{2}P_{ij}^2}{\sqrt{(\delta_{i,h}^{2}g_{i,h}+\frac{1}{2}g_{i,h}^{2})(\delta_{j,h}^{2}g_{j,h}+\frac{1}{2}g_{j,h}^{2})}} \\
 = & \frac{ \dfrac{h^4}{n^4}\left[ \dfrac{2\sigma_i^2}{\theta^2} \left\{ \dfrac{ \sigma_i^2}{\theta^2}   - \left(\sigma_i^{(1)}\right)^2  - \delta_{i,h} \delta_{j,h} \frac{n^2}{h^2}\right\} + \left\{\delta_{i,h} \delta_{j,h} \frac{n^2}{h^2} +  \dfrac{1}{2}\left( \sigma_i^{(1)}\right)^2 \right\}\left( \sigma_i^{(1)}\right)^2  + o\left(n^{-1} \right)\right] }
{\sqrt{\left(\delta_{i,h}^{2}g_{i,h}+\frac{1}{2}g_{i,h}^{2}\right)\left(\delta_{j,h}^{2}g_{j,h}+\frac{1}{2}g_{j,h}^{2}\right)} }\\
 = & \frac{O(n^{-4})}{O(n^{-2\alpha})} = O(n^{-2(2-\alpha)}).
\end{align*}
Note that the correlation between the squared pseudo-residuals $D_{i,h}^{2}$ and $D_{j,h}^{2}$ for $i \neq j$ converges to 0 with the infill asymptotic. Simple differencing not only removes the feature of a mean function but also drastically reduces correlation in the data the closer $\alpha$ is to 2. When $\delta_{i,h}=o\left(n^{-1}\right)$ is negligible for all $i = 1, \dots, n-h$, in other words $\mu(s)\in \Lambda_{q}$ where $q < 1$, the third line of equality for $cor(D_{i,h}^{2}, D_{j,h}^{2})$ is reduced to
\[ cor(D_{i,h}^{2}, D_{j,h}^{2})
= \dfrac{h^4}{n^4} \dfrac{\left\{\dfrac{2\sigma_i^2}{\theta^2} - \left( \sigma_i^{(1)}\right)^2 \right\}^2 + o\left(n^{-1}\right)}{ g_{i,h} g_{j,h}}.
\]
Since $g_{i,h} = o\left(n^{-\alpha}\right)$, the rate of convergence for the correlation is again $O\left(n^{-2(2-\alpha)}\right)$ where $0< \alpha < 2$. 

\subsection{Asymptotic result}
\label{sec:Risk}
The point-wise risk of the local variogram estimator is the sum of the squared bias in (\ref{eq:bias0}) and the variance in (\ref{eq:var0}). The asymptotic point-wise risk can be derived from combining the results of Theorems \ref{theorem1}  and  \ref{theorem2}.
\begin{theorem}
\label{theorem3}
Consider estimating the variance function of a one-dimensional nonstationary process with $n$ equally-spaced observations, whose data model follows (\ref{eq:data-model}) and (\ref{eq:rho}). We assume that $\mu(s)\in \Lambda_{q},\, q\geq 0$, $\sigma^{2}(s)\in \Lambda_{\beta},\,\beta\geq 2$ and that the bandwidth $\lambda=O\left(n^{-x} \right)$ where $0<x<1$. When the order of Gasser-M\"uller kernel $m$ is greater than both $q$ and $\beta$, the point-wise risk of the estimator of local variogram in (\ref{eq:variog-est}) and the asymptotic convergence rate of bandwidth are
\begin{equation}
\label{eq:risk1}
Risk(\hat{\gamma}_{\lambda}(s,h))  = 
\begin{cases} 
O\left(n^{-4q}\right) &  \mbox{ where } \quad \lambda   \asymp n^{-1-2\alpha+4q}\\
O\left( n^{-4}\right) & \mbox{ where } \quad   \lambda   \asymp n^{3-2\alpha}
\end{cases}
 \end{equation} 
given $ \alpha < 2q < \min(\alpha + \dfrac{1}{2}, 2)$ for the top case 
and $q\geq1$ and $\alpha > \dfrac{3}{2}$ for the bottom case. When the order of Gasser-M\"uller kernel $m$  is greater than only one of $q > 1$ or $\beta$, the point-wise risk is
\begin{equation}
\label{eq:risk2}
Risk(\hat{\gamma}_{\lambda}(s,h))  = 
\begin{cases}
 O\left(n^{-2m (1+2\alpha)/(1+2m)} \right) &\mbox{where } \quad \lambda   \asymp n^{-(1+2\alpha)/(1+2m)}\\ 
 O\left(n^{-2\alpha-2m/(1+2m)}  \right) & \mbox{where }\quad \lambda   \asymp  n^{-1/(1+2m)}
\end{cases} 
\end{equation}
given $\alpha < \min\left(2q, \dfrac{3}{2}\right)$ for the top and $\ \alpha <  2q$ for the bottom.
\end{theorem}

%
\begin{proof} The asymptotic bias and variance are derived in (\ref{eq:bias1}) and (\ref{eq:var1}) respectively.
{\allowdisplaybreaks 
\begin{align}
\label{eq:a-risk}
Risk&\left(\hat{\gamma}_{\lambda}(s,h),\gamma(s,h)\right) 
 = bias(\hat{\gamma}_{\lambda}(s,h))^{2}+var\left(\hat{\gamma}_{\lambda}(s,h)\right) \nonumber \\
&= \begin{cases} 
O\left(n^{-4} + n^{-4q} \right)+ O\left(\dfrac{1}{n\lambda}\right) O\left( n^{-2\alpha} \right)  \\  \hspace{3.5in}   \mbox{ where } \quad  q, \beta < m \\
O\left(n^{-4} + n^{-4q}  +  n^{-2\alpha} \lambda^{2m} \right)+ O\left( \dfrac{1}{n\lambda} \right) O\left( n^{-2\alpha} + n^{-2q-\alpha}\right)\\ 
\hspace{3.5in}   \mbox{ where } \quad q < m \leq \beta,\\
O(n^{-4} + n^{-4q} + \lambda^{2 m}) + O\left(\dfrac{1}{n\lambda}\right) O\left( n^{-2\alpha} + n^{-2q-\alpha}\right) \\
\hspace{3.5in} \mbox{ where } \quad  m \leq q.
\end{cases} 
\end{align}}

\begin{enumerate}
\item[A.] Assume that $m >q$ and $m> \beta$.  

	\begin{enumerate}  
	\item[(i)] When $q\geq 1$, $\alpha < 2q$ holds true because $0<\alpha<2$, and (\ref{eq:var1}) reduces to $O\left(n^{-2\alpha-1} \lambda^{-1}\right)$. When the asymptotic bias is $O(n^{-2})$, the bandwidth condition is met and $\lambda \asymp n^{3-2\alpha}$, which suggests $\frac{3}{2}  < \alpha$.

	\item[(ii)] When $\frac{1}{2}<q<1$, 
the asymptotic order of bias is $O\left(n^{-2q} \right)$. With  $\alpha < 2q$, the asymptotic variance of $O\left(n^{-2\alpha -1} \lambda^{-1}\right)$. Then, $\lambda \asymp n^{-1-2\alpha +4q}$.
\end{enumerate}

\item[B.] Assume $q < m \leq \beta$.
	\begin{itemize}
	\item When $\alpha < 2q$, the asymptotic variance is $O\left(n^{-2\alpha-1} \lambda^{-1}\right)$. 
	\item[(i)]   The bias is $O(n^{-\alpha}\lambda^{m})$ when $\alpha < 2q$, and it gives $\lambda   \asymp n^{-1/(1+2m)}$.
	\item[(ii)] The bias of $O(n^{-2})$ has similar conditions as A(i) where $q>1$, $\alpha > 2 - \dfrac{m}{1+2m}$, and $\lambda  \asymp n^{3-2\alpha}$.
	\item[(iii)]  The bias of $O\left( n^{-2q}\right)$ has similar conditions as A(ii) where $q\leq1$, $\alpha >\dfrac{1}{ 2q - \frac{m}{1+2m}}$, and  $\lambda  \asymp n^{-1-2\alpha +4q}$.  
	\item When $\alpha \geq 2q$, the asymptotic variance is $O\left(n^{-2q -\alpha-1} \lambda^{-1}\right)$. 
	\item[(iv)]   The bias is $O(n^{-\alpha}\lambda^{m})$ when $\alpha < 2q$, which is a contradiction to $\alpha \geq 2q$.
	\item[(v)] The bias  is $O(n^{-2})$ when $q > 1$, which is a contradiction to $0< \alpha < 2$ with $\alpha \geq 2q$.
	\item[(vi)]  The bias of $O\left( n^{-2q}\right)$ suggests $\lambda   \asymp n^{-(1+\alpha-2q)} = o(n^{-1})$, which is a contradiction to $\lambda = O(n^{-1})$.  
	\end{itemize}

\item[C.] Assume $m \leq q$.
	\begin{itemize}
		\item[(i)]   Assuming that the bias is $O(\lambda^{m})$, we have $\lambda \asymp n^{-(1+2\alpha)/(1+2m)}$  when $\alpha < 2q$; and $\lambda   \asymp n^{-(1+\alpha+2q)/(1+2m)}$ when $\alpha \geq 2q$. Checking the assumptions requires:
		\begin{itemize}
		\item[-]  $O(\lambda^{m}) > O(n^{-2}) \Longleftrightarrow \dfrac{m(1+2\alpha)}{1+2m} < 2$. This implies $m < \dfrac{2}{2\alpha - 3}$ when $ \alpha > \dfrac{3}{2}$. 
		\item[-]  $O(\lambda^{m}) > O(n^{-2q}) \Longleftrightarrow \dfrac{m(1+2\alpha)}{1+2m} < 2q$. This holds true for the given condition $m\leq q$ since $ \dfrac{1+2\alpha}{1+2m}< \dfrac{1+4q}{1+2m} < \dfrac{2q}{m}$ where the second inequality holds when $m < 2q$. 		
	\end{itemize}

	\item[(ii)] The bias of $O(n^{-2})$ has the same case as A(i) and B(ii), and $\lambda  \asymp n^{3-2\alpha}$.
	\item[(iii)]  By the second check in C(i), the bias cannot be $O\left( n^{-2q}\right)$.  
	\item[(v)] The bias of $O(n^{-2})$ does not hold as in B(v).
	\end{itemize} 
\end{enumerate}
\end{proof}

\begin{remark}
In (\ref{eq:risk1}) of Theorem \ref{theorem3}, the order of risks and bandwidth are continuous at $q=1$. In (\ref{eq:risk2}) of Theorem \ref{theorem3}, the risks in $O\left(n^{-2m(1+2\alpha)/(1+2m)} \right)$ and $O\left(n^{-2\alpha - 2m/(1+2m)}\right)$ are quite similar. In fact $\dfrac{O\left(n^{-2m(1+2\alpha)/(1+2m)}\right)}{O\left( n^{-2\alpha - 2m/(1+2m)} \right)} = O\left(n^{2\alpha/(1+2m)}\right) = o(n^{1})$.  In both cases, the implied range of the smoothness parameter is $\alpha <3/2$.
\end{remark}
\begin{remark}
There is divergence of risk when (i) $q \geq \beta$ and/ or (ii) the process is very smooth with $\alpha \gtrsim 3/2$. In both cases a variance function is masked by a mean process. 
\end{remark}
Given $m = \beta$, as $\alpha \rightarrow 0$  ($\alpha = 0$ suggests an independent process), the risk converges to $O\left(n^{-2\beta/(1+2\beta)} \right)$ in all three cases. This result is consistent with the nonparametric estimation of $\beta$ differentiable functions. 

\section{Bandwidth Selection}
\label{sec:algorithm}
We are interested in estimating the variance function embedded in a nonstationary spatial process where $\beta > q$, i.e. the variance function has a greater differentiability than the mean function, and the standardized spatial process is isotropic. The estimation of a variogram at a fixed lag-$h$ is given as the average of squared differences at that lag, and the estimation of a local variogram at a fixed location and distance is a local average of the squared differences for the given location. The concept of ``local'' or nearby locations can be defined by a bandwidth of a smoothing kernel, and this section discusses the bandwidth selection procedure.

It is well known in nonparametric statistics literature that when the underlying data contain correlation in additive errors, cross-validation requires an adjustment in data or in the penalty term of objective functions. \cite{Altman:1990fk} proposed to adjust the weights of the correlated residuals. \cite{Han:2008tg} added  a penalty term to the likelihood function to adjust for the correlation and then simultaneously estimated the bandwidth and the correlation parameters. \cite{Opsomer:2001qe} compiled several proposals of bandwidth selection in nonparametric regression with correlated errors and addressed recent developments on the theoretical front. 
 We choose leave-one-out cross-validation to minimize the mean square prediction errors of local variogram.


Recall that $D_{i.h}^2$ denotes the $i$\textsuperscript{th} squared difference of lag-$h$ process. Let $d_{i,h}^2$ represent a realization of $D_{i,h}^2$, and define a raw deviance of local variogram  estimation at $s_{i}+h/2$ as 
\begin{equation}
\hat{\epsilon}_{i}=d_{i,h}^2-\hat{\gamma}_{L}\left(s_i + \frac{h}{2}, h\right). 
\label{eq:epsilon}
\end{equation}

Let the covariance matrix of the deviances be $C_\epsilon$ whose  $(i,j)$ element is $cov\left(\epsilon_{i}, \epsilon_{j}\right)$. We use it to de-correlate the raw deviances, $resid_{\epsilon} = (\epsilon_1,\cdots, \epsilon_{n})$, of the local variogram estimation and denote the de-correlated residuals as
\[\xi = C_{\epsilon}^{-1/2}\widehat{resid}_{\epsilon}.\]
The choice of a covariance model and the parameter values are not sensitive to bandwidth estimation because the correlation in $resid_{\epsilon}$ is weak. We assume an exponential covariance model with a range parameter $\phi = 0.01$ as follows. $cov\left(\epsilon_{i}, \epsilon_{j}\right) = \exp\left(-\frac{1}{\phi}\frac{|i-j|}{n}\right)$. The algorithm for bandwidth selection is below:


\begin{enumerate}
	\item Use a differencing lag $h=1$. Create a set of bandwidths $\left\{ \lambda_{j}\right\}$ whose value would not exceed 1/2 the range of the sample domain. Estimate the local variogram over a given domain using equation (\ref{eq:variog-est}) for each $\lambda_{j}.
	$\footnote{An important consideration in both local variogram or variance function estimation is to guarantee that they are non-negative or positive everywhere. When a bandwidth is small, the smoothing may result in negative values most often near the boundaries. In order to address this problem, one may fix the bandwidth size for locations near the boundary.}
	
	\item Calculate $\widehat{resid}_{\epsilon}$ in (\ref{eq:epsilon}) for each $\lambda_{j}$.

	\item Select a bandwidth using leave-one-out cross-validation, whose minimization of the overall mean-squared-error is approximated   as follows.  
\[\displaystyle{\hat{\lambda} \leftarrow \arg_\lambda \min \sum_{i=1}^{n-h} \left( \frac{\xi_i}{1-M_{(i,i)}} \right)}^2\]
Note $M$ is an $(n-h) \times (n-h)$ smoothing matrix of $D_{i,h}^2$ and and the ($i,\,i$) element of $M$ is $M_{\left(i,i\right)}=K(0)$.

	\item Use the bandwidth from 3 to estimate the local variogram and scale the de-trended original data by the square root of the estimation. Here, we assume $\big\{Z_{i} \big\}_{i=1}^{N}$ is the unaccounted, correlated variation in the process.
	\[\Big\{ Z_{i}^* \Big\}_{i=1}^{N} \leftarrow \dfrac{\big\{Z_{i} \big\}_{i=1}^{N} }{\Big\{ \sqrt{\widehat{\gamma}_{L, \hat{\lambda}} (s_i, h)}\Big\}_{i=1}^N }\]

	\item Select any correlation model for $\big\{ Z_{i}^* \big\}$ and estimate the correlation parameters $\theta$. Resize the local variogram by $\hat{\sigma}^2_{*}$ estimated from $\big\{ Z_{i}^* \big\}$.

	\item Plug-in the correlation model and the parameters to estimate variance at any $s$: 
	\[ \hat{\sigma}^{2}(s) \leftarrow \dfrac{ \widehat{\gamma}_{L, \hat{\lambda}}(s;h)\,\hat{\sigma}^2_{*}}{1-\hat{\rho}(h; \hat{\theta})}\]
\end{enumerate}

A simulation study result is in the next section, and it includes the result of bandwidth selection. Since bandwidth selection is necessary in local estimation of functions, it is important to use computationally inexpensive approaches. The generalized cross-validation we use is an approximation of cross-validation score that is not very heavy on computing in comparison to a leave-one-out cross-validation. In \cite{Anderes:2011nx}, the bandwidth selection adds greater computing complexity in addition to the functional estimation since the proposal is simulation-based and requires either inverting a large matrix or taking derivatives of an estimated function for every iteration of simulation. We also need a matrix inversion for de-correlation, but it is set up to perform once for the functional estimation. Also, there is no need to iterate  from 3-6, which could be a needed optimization process, because the estimation of correlation parameter is stable across a wide range choices in constructing $C_{\epsilon}$.


.

\section{Simulation Study}
\label{sec:simulation}

Here we compare the difference-based method and the likelihood-based method in terms of statistical and computing efficiencies. We also examine the effect of dependence in correlated errors on functional estimations. We define \textit{oracle bandwidth} as the bandwidth that yields the minimum discretely integrated mean square error ($DMSE$) of the estimated function. To provide equal footing on the difference- and likelihood-based estimations, we assume that the correlation functions and the parameter values are known. We label the oracle bandwidths under such set-up as `Diff-$\lambda^{O}$' and `Like-$\lambda^{O}$'. In applying the proposed method as described in Section \ref{sec:algorithm}, we select a bandwidth `Diff-$\lambda^*$' and estimate correlation parameters.  \cite{Levine:2006cr} also proposes a bandwidth selection method for heteroskedastic but independent errors, and we refer to them as `Levine-$\lambda$'.

\label{sec:simulation}
\subsection{Set-up}
\label{sec:simulation-setup}
Assume a data model $Z(s) = \mu(s)+\sigma(s)X(s)$ as in (\ref{eq:data-model}) and set $\mu(s)=0$ to test the method directly on the correlated error processes or under a constant mean assumption. We fix the stationary error process $\{X(s)\}$ as a Gaussian process for analytical tractability. They are easy to simulate, and the likelihood-based approach should prefer tractable likelihood functions and provide no disadvantage when compared to the difference-based estimation. The dependent structure is generated using an exponential correlation function with a range parameter set at three levels $\theta$=0.1,  0.01  and 0. The latter, in fact, refers to an independent error setting. The processes are generated on an equally spaced grid over a unit interval, $0\leq s \leq1$. Four sample sizes $n=100$,  200, 500, and 1000 are used. The standard deviation functions are chosen to examine the effect of differentiability of the mean versus the variance functions especially on the bandwidth selections. Here is the summary of experimental details. Note that \cite{Anderes:2011nx}  have used $\sigma(s)$ in 3(a), and we added 3(b).

\begin{enumerate}
\item $n$ = 100,  200, 500 and 1000

\item $\sigma(s):[0,1]\rightarrow\mathbf{R}^{+}$  and set $s\in\Big\{0,\frac{1}{n-1}, \dots, \frac{n-2}{n-1}, 1\Big\}.$

(a) an infinitely-differentiable function: $\sigma(s)=2\sin(s/0.15)+2.8,$

(b) a step function:  $\sigma(s)=1+\mathbbm{1}_{\{1/3<s\leq1\}}.$

\item For a stationary error processl $\left\{ X_{s}\right\}$, let $cor\big(X (s), X(s+h/n) \big) = \exp\left( -\frac{1}{\theta} \frac{h}{n} \right)$  and set $\theta=0.1$, 0.01 where $h$ is small and $0\leq s \leq1-\frac{h}{n}$. Also, set $cor\big(X (s), X(s') \big) = 0$ for any $0\leq s, s' \leq 1$.

\item Draw 100 random samples of each random procress. 
\end{enumerate}

We define $DMSE(\hat{\sigma}^{2}_{\hat{\lambda}})= \sum_{i=1}^{n}\{\hat{\sigma}_{i, \hat{\lambda}}-\sigma_{i}\}^{2}/n$ and $L_{\infty}(\hat{\sigma}^{2}_{\hat{\lambda}})$ $ = \max_{i}\left\{ |\hat{\sigma}_{i,\hat{\lambda}}^{2}-\sigma_{i}^{2}|\right\}$. We estimate the variance functions at 100 equally spaced locations on [0,1]. Then, we evaluate the estimated functions using discretely integrated mean square error ($DMSE$)  as an overall measure of functional estimation and the maximum absolute deviation, i.e. supremum norm $L_{\infty}$, to represent the worst estimation.


\begin{figure}
\begin{centering}
	\includegraphics[width=0.7\textwidth]{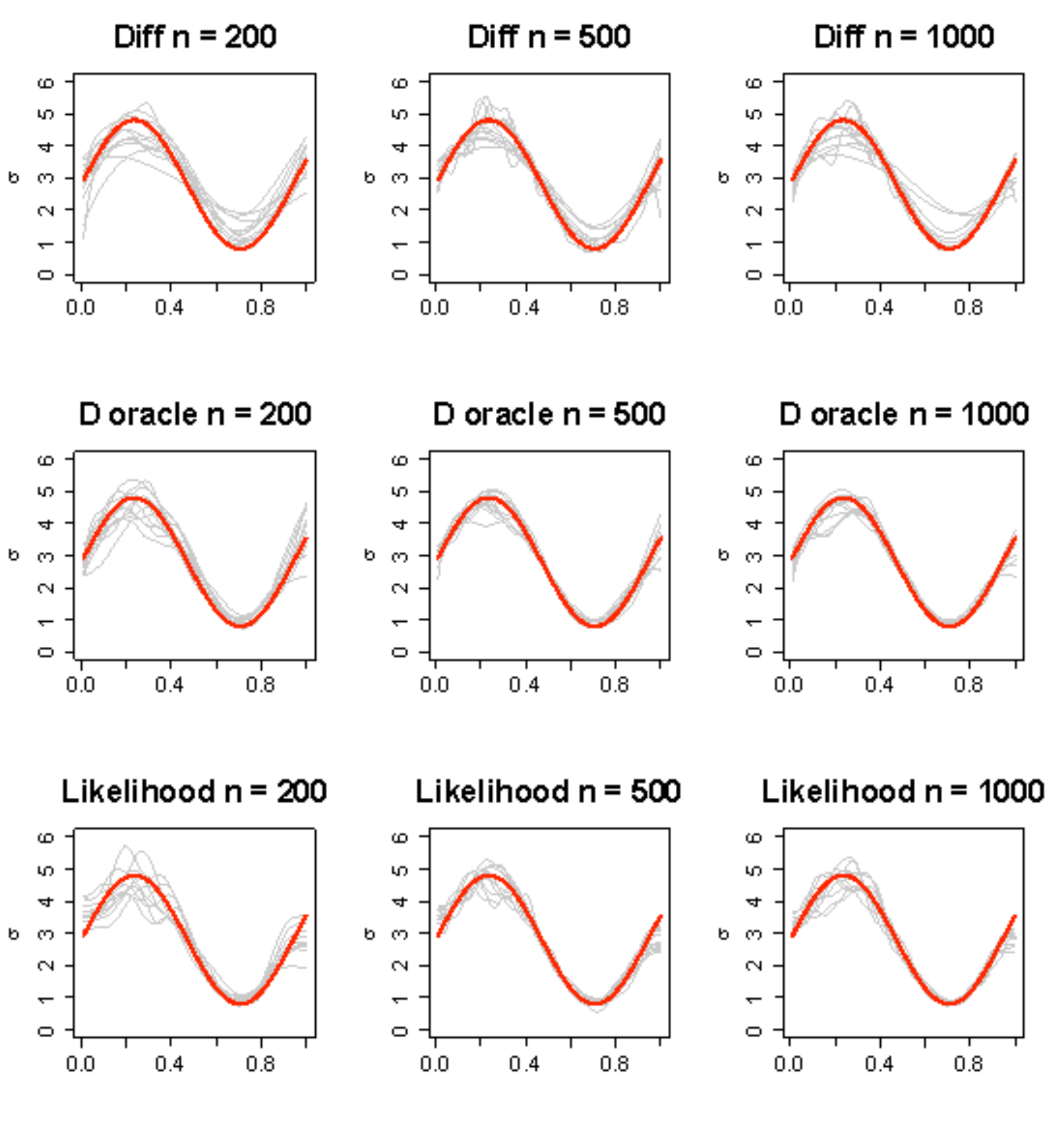}
        \caption{\label{fig:sine} The true standard deviation function is in thick red line. The estimation results are in thin gray lines for three estimation methods crossed with three levels of sample size.}
\end{centering}
\end{figure}

\subsection{Results}

Figure \ref{fig:sine} shows a few variance function estimation results where the true  $\sigma(s)$ is sinusoidal. The thick red line represents the true function and the thin gray lines are estimation results of several simulations. The first row implemented our proposed method; the second row applied the same idea with oracle bandwidths and known covariance parameters; and the third used \cite{Anderes:2011nx}'s likelihood-based method with oracle bandwidths and covariance parameters.  As expected in an infill design, the estimation becomes more precise as the number of points $n$ increases. The likelihood-based functional estimation shows more undulation than the difference-based estimation results. In other words, the oracle bandwidths for the likelihood-based method are underestimated. They were selected to minimize the discretized mean square error ($DMSE$), and therefore smoothing out the undulation should result in larger discretely integrated mean square errors than from the current functional estimations. 

\begin{figure}
\begin{centering}
	\includegraphics[width=0.75\textwidth]{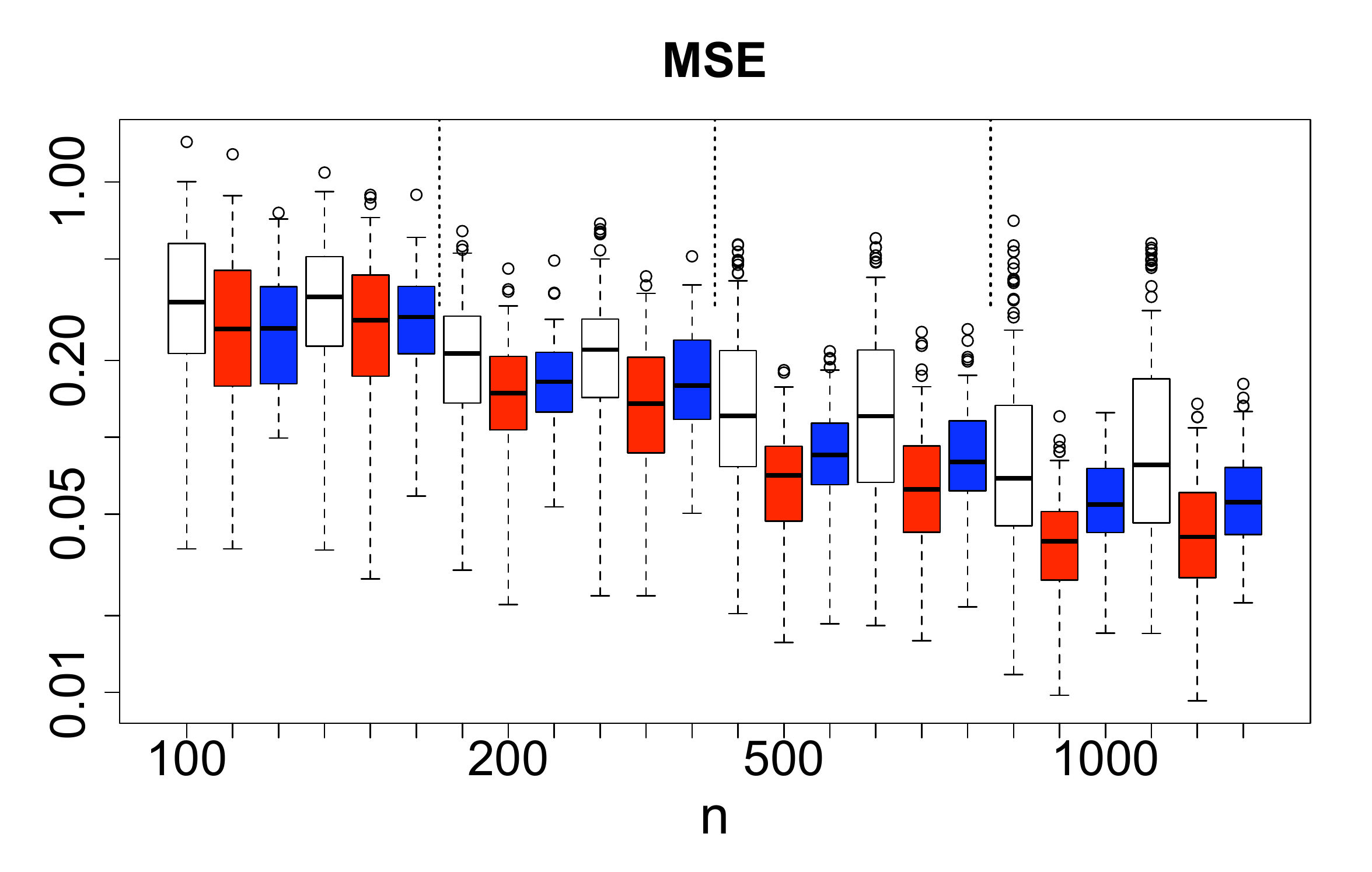}
        \caption{\label{fig:DMSE} Comparing  the proposed method (white boxplots), difference-based method with oracle bandwidths (red), and the likelihood-based method with oracle bandwidths (blue) for the estimation of sinusoidal $\sigma(\cdot)$ using $DMSE$.}
        	\includegraphics[width=0.75\textwidth]{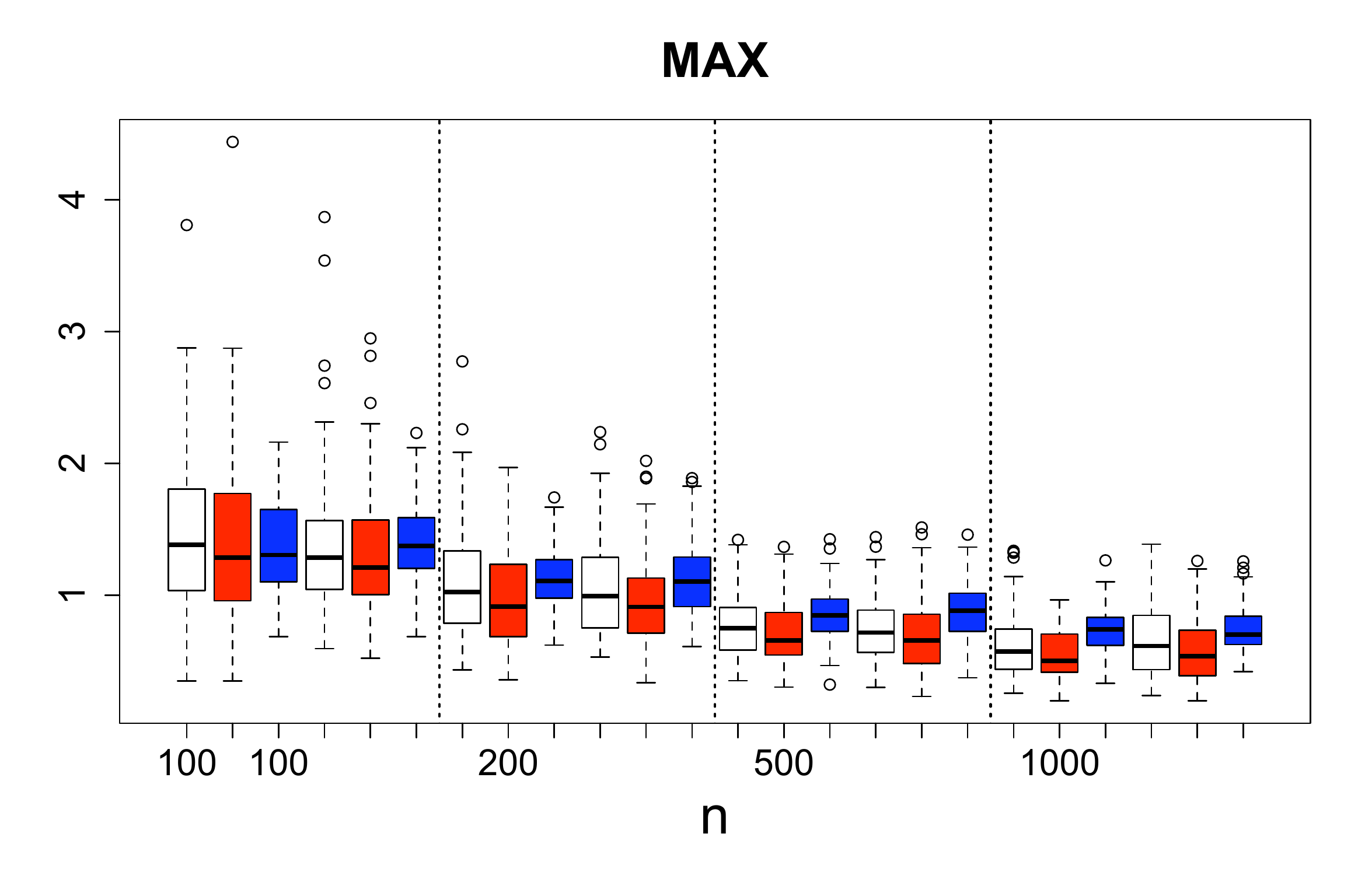}
        \caption{\label{fig:MAX} Comparing the results of functional estimation using $MAX$, the $L_{\infty}$ norm. }        
\end{centering}
\end{figure}

Figures \ref{fig:DMSE} and \ref{fig:MAX} summarize the simulation results of the sinusoidal $\sigma(\cdot)$ with discretely integrated mean square error ($DMSE$) and the maximum absolute deviation ($MAX$) respectively. The color of the boxplots represents the estimation method where the proposed method is in white, the proposed with $\lambda^{O}$ in red, and the likelihood-based method with $\lambda^{O}$ in blue. There are two sets of triad-colored boxplots for each sample size where the first set indicates weakly correlated processes (with $\theta=0.01$) and the second set strongly correlated processes (with $\theta=0.1$). The vertical dashed lines demarcate the change in sample sizes from 100, 200, 500, to 1000.  

There is little difference between \textit{Diff} and \textit{Likelihood}-based methods when the oracle bandwidths are plugged in and where $n$ is less than 200. As $n$ grows, \textit{Diff} has smaller $DMSE$ and $MAX$ than \textit{Likelihood}.  The differences are attributable  to the  under-smoothed functional estimations and to the boundary effects as the likelihood-based functional estimations are flat near the boundaries (as shown in the third row of Figure \ref{fig:sine}). Note that the summary measures show a reasonable range of values, considering that the functional value of $\sigma(\cdot)$ ranges from  0.8 - 4.8; the $DMSE$s are mostly less than 0.5 and the $MAX$s are generally less than 1.5. The comparison of two methods via Figures \ref{fig:DMSE} and \ref{fig:MAX} shows that \textit{Diff} is a simpler method with lower risk in estimation than \textit{Likelihood}.

 


\linespread{1}
\begin{table}
\def~{\hphantom{0}}
\caption{Bandwidth selection summary for (a) sine and (b) step $\sigma(\cdot)$ function estimation. Oracle bandwidths, $\lambda^{O}$, achieve the minimum DMSE for both difference- and likelihood-based methods, our proposed bandwidth selections are represented as $\lambda*$, and there is Levine's method.}{
\begin{tabular}{c | c | ccc | ccc}
& Bandwidth & \multicolumn{3}{c}{(a) Sine} & \multicolumn{3}{c}{(b) Step} \\ \hline
$n$ & Methods & $\theta=0.1$ & $\theta=0.01$	& indep.	& $\theta=0.1$ & $\theta=0.01$ & indep. \\ \hline
100 &  Diff-$\lambda^{O}$	 & 0.203 	&0.206 	&0.209 		& 0.218	& 0.222	& 0.229\\
	&&	\scriptsize(.054)	& \scriptsize(.059)	& \scriptsize(.052)	&	\scriptsize(.071)	& \scriptsize(.084)	& \scriptsize(.076)\\
	& Diff--$\lambda*$ & 0.262 &	0.281 &	0.266 	& 0.405	& 0.415	& 0.434\\
	&&	\scriptsize(.074)	& \scriptsize(.079)	& \scriptsize(.069)	&	(.126)	& \scriptsize(.087)	& \scriptsize(.074)\\
	& Levine	& 0.356	& 0.455	& 0.420		& 0.360	& 0.467	& 0.418\\
	&&	\scriptsize(.297)	& \scriptsize(.274)	& \scriptsize(.281)	&	\scriptsize(.304)	& \scriptsize(.267)	& \scriptsize(.289)\\
	& Like--$\lambda^{O}$	& 0.165 	& 0.168 	& 0.154 		& 0.137	& 0.138	& 0.133\\
	&&	\scriptsize(.054)	& \scriptsize(.055)	& \scriptsize(.033)	&	\scriptsize(.032)	& \scriptsize(.030)	& \scriptsize(.030)\\ \hline
200	& Diff-$\lambda^{O}$	& 0.170 	& 0.171 	& 0.177 		& 0.191	& 0.185	& 0.203\\
	&&	\scriptsize(.034)	& \scriptsize(.037)	& \scriptsize(.046)	& \scriptsize(.050)	& \scriptsize(.060)	& \scriptsize(.066)\\
	& Diff--$\lambda*$ &	0.240 &	0.218  &	0.190 	&	0.381 &	0.336 &	0.289\\
	&&	\scriptsize(.090)	& \scriptsize(.108)	& \scriptsize(.119)	&	\scriptsize(.126)	& \scriptsize(.143)	& \scriptsize(.163)\\
	&  Levine	& 0.234	& 0.380	& 0.347		& 0.248	& 0.369	& 0.334\\
	&&	\scriptsize(.248)	& \scriptsize(.224)	& \scriptsize(.229)	&	\scriptsize(.249)	& \scriptsize(.230)	& \scriptsize(.217)\\
	& Like--$\lambda^{O}$	& 0.131 	& 0.129 	& 0.127 		& 0.113	& 0.113	& 0.112\\
	&&	\scriptsize(.034)	& \scriptsize(.028)	& \scriptsize(.021)	&	\scriptsize(.025)	& \scriptsize(.024)	& \scriptsize(.023)\\ \hline
500	& Diff-$\lambda^{O}$	& 0.140 	& 0.141 	& 0.154 		& 0.154	& 0.152	& 0.158\\
	&&	\scriptsize(.027)	& \scriptsize(.031)	& \scriptsize(.037)	& 	\scriptsize(.042)	& \scriptsize(.042)	& \scriptsize(.047)\\
	& Diff--$\lambda*$h &	0.217 &	0.205 &	0.180 	&	0.357 &	0.329 &	0.260\\
	&&	\scriptsize(.107)	& \scriptsize(.117)	& \scriptsize(.111)	&	\scriptsize(.143)	& \scriptsize(.147)	& \scriptsize(.159)\\
	&  Levine	& 0.186	& 0.256	& 0.232		& 0.192	& 0.264	& 0.240\\
	&&	\scriptsize(.186)	& \scriptsize(.164)	& \scriptsize(.165)	&	\scriptsize(.193)	& \scriptsize(.152)	& \scriptsize(.166)\\
	& Like--$\lambda^{O}$	& 0.098 	& 0.098	& 0.100		& 0.091	& 0.090	& 0.094\\
	&&	\scriptsize(.016)	 & \scriptsize(.016)	& \scriptsize(.016)	&	\scriptsize(.019)	& \scriptsize(.016)	& \scriptsize(.017)\\ \hline
1000	 & Diff-$\lambda^{O}$	& 0.120 	& 0.121 	& 0.133 		& 0.131	& 0.125	& 0.148\\
	&&	\scriptsize(.026)	& \scriptsize(.026)	& \scriptsize(.023)	&	\scriptsize(.033)	& \scriptsize(.033)	& \scriptsize(.038)\\
	& Diff--$\lambda*$ &	0.209 &	0.186 &	0.170 	&	0.329 &	0.300 &	0.255\\
	&&	\scriptsize(.121)	& \scriptsize(.117)	& \scriptsize(.109)	&	\scriptsize(.159)	& \scriptsize(.157)	& \scriptsize(.165)\\
	&  Levine	& 0.180	& 0.289	& 0.174		& 0.199	& 0.288	& 0.191\\
	&&	\scriptsize(.155)	& \scriptsize(.118)	& \scriptsize(.094)	&	\scriptsize(.157)	& \scriptsize(.123)	& \scriptsize(.092)\\
	& Like--$\lambda^{O}$	& 0.086 	& 0.084 	& 0.086 		& 0.078	& 0.076	& 0.078\\
	&&	\scriptsize(.013)	& \scriptsize(.011)	& \scriptsize(.013)	&	\scriptsize(.015)	& \scriptsize(.013)	& \scriptsize(.014)\\ \hline
\end{tabular}}
\label{tb:band}
\end{table}

 
 As noted in (\ref{eq:risk1}) of Theorem \ref{theorem3}, the greater differentiability a variance function has, the quicker the risk converges. To confirm the theoretical results of asymptotic risk,  the simulation study involved four variance functions changing the differentiability, but we omit full summaries. The strength of the dependency in the process, that is, the correlation at a fixed distance, also does not affect the asymptotic risk. When the true values ($\theta = 0.01$ and 0.1)  of the range parameter were plugged in for the \textit{Diff} (the red) and \textit{Likelihood} (the blue boxplots) showed very similar estimation results in both $DMSE$ and $MAX$. In practice, the covariance parameter estimation brings uncertainty to the estimation as we see a wider range of $DMSE$ and $MAX$ in the white boxplots.

Table \ref{tb:band} contains the summary of selected bandwidth of sinusoidal and step $\sigma(\cdot)$ function. We used a degree 6 Gasser-M\"uller kernel for the differenced-based method and a Gaussian-based higher order kernel for the likelihood-based method. We see that the size of the ``oracle bandwidth" is slightly smaller  for a dependent process than an independent process when $n=500$ and 1000. Our bandwidth selection gives the opposite result in that the independent error process gets the smallest average. Since the range of selected bandwidths is wide, there are under-smoothed and over-smoothed functions in the top row of Figure \ref{fig:sine} especially where $n$ is large.

We have shown through a simulation study that difference-based estimation has a smaller  $DMSE$ than a likelihood-based approach. In nonparametric regression, boundary bias can be easily fixed by adjusting the objective function near the boundary, whereas for the likelihood-based method generalized estimating equations are suggested. Another contrast between the two approaches is in computing time. A difference-based method needs no matrix inversion and reduces the computing time by $O(n^{-2})$  to that of a likelihood-based method, where $n$ is the length of the data process. The bandwidth selection idea by \cite{Anderes:2011nx} also requires a global covariance matrix inversion and increases the computing time by $O(mn^{2})$  where $m$ is the number of simulations for generating a globally stationary process to test against the observed nonstationary process. While their bandwidth selection ideas are insightful and useful when there is a specific data model that can be simulated, it is much more costly to perform likelihood-based estimation in terms of computing time and power.

\section{Summary and Future Work}
\label{sec:conclusion}
We have developed a nonparametric variance function estimator for a one-dimensional nonstationary process whose stationary correlation structure is isotropic. Under certain regularity conditions we can directly estimate a variance function, applying a difference filter to the data. We assume that the error processes are additive. The mean can be estimated and then removed from the data to employ our method and estimate the variance function. A direct application of the method to compute pseudo-residuals is possible assuming a smoothly varying mean function, and this should reduce the bias caused from estimating the mean function. We have investigated infill asymptotic  properties of the local variogram estimator and have shown that the asymptotic rate of convergence is dependent on the relative smoothness of mean function to the smoothness of variance function and the mean square differentiability of the data process.


We would like to extend the difference-based method to a two-dimensional random field nonstationary variance function estimation. In such setting the number of difference filter choices increases in shape and size, and the dependence would be stronger among the filtered process leading to new challenges for estimation.


\bigskip
\begin{center}
{\large\bf SUPPLEMENTARY MATERIAL}
\end{center}


Here is the detailed expansion of the variance of $h$-lagged nonstationary process with smooth mean and variance function. This details (\ref{eq:deriving-locvariog}) in deriving the local variogram (\ref{eq:Def-locvariog}) as the main term of the expansion.
{\allowdisplaybreaks 
\begin{align}
\label{eq:variogram-full}
 var&\left(Z\left(s-\frac{h}{2n}\right)-Z\left(s+\frac{h}{2n}\right)\right) \nonumber \\
 = & 2\left(\sigma^{2}(s)+\frac{\sigma^{2(2)}(s)}{2!}\left(\frac{h}{2n}\right)^{2}+\frac{\sigma^{2(4)}(s)}{4!}\left(\frac{h}{2n}\right)^{4}+o\left(\left(\frac{h}{2n}\right){}^{5}\right)\right) \nonumber \\
& -2\rho_{h}\sum_{k=0}^{p}\left\{ \left(\frac{\sigma^{(k)}(s)}{k!}\right)^{2}\left(\frac{h}{2n}\right)^{2k}(-1)^{k}+2\sum_{i+j=2k,\, i\neq j}\frac{\sigma^{(i)}(s)}{i!}\frac{\sigma^{(j)}(s)}{j!}\left(\frac{h}{2n}\right)^{2k}\right\}  \nonumber \\
 = & 2(1-\rho_{h})\left\{ \sigma^{2}(s)+\frac{\sigma^{2(2)}(s)}{2!}\left(\frac{h}{2n}\right)^{2}+\frac{\sigma^{2(4)}(s)}{4!}\left(\frac{h}{2n}\right)^{4}+o\left(\left(\frac{h}{2n}\right){}^{5}\right)\right\} \nonumber \\
& +\rho_{h} \left[\frac{\left(\sigma^{2(1)}(s)\right)^{2}}{\sigma^{2}(s)}\left(\frac{h}{2n}\right)^{2} \right. \nonumber \\
& \quad + \left.\left\{ \frac{\left(\sigma^{2(1)}(s)\right)^{4}}{32\left(\sigma^{2}(s)\right)^{3}}+\frac{\left(\sigma^{2(1)}(s)\right)^{2}}{8\left(\sigma^{2}(s)\right)^{2}}-\frac{3\left(\sigma^{2(2)}(s)\right)^{2}}{8\sigma^{2}(s)}+\frac{\sigma^{2(1)}(s)\sigma^{2(3)}(s)}{6\sigma^{2}(s)}\right\} \left(\frac{h}{2n}\right)^{4}\right] \nonumber \\
 = & 2\sigma^{2}(s)(1-\rho_{h})+\left\{ \sigma^{2(2)}(s)(1-\rho_{h})+\frac{\left(\sigma^{2(1)}(s)\right)^{2}}{\sigma^{2}(s)}\rho_{h}\right\} \left(\frac{h}{2n}\right)^{2}+o\left(\left(\frac{h}{2n}\right)^{3}\right). 
\end{align}}


$\delta_{i,h} \leq c_{\mu}^{2}(h/n)^{q}$.
Under the condition that $\mu(\cdot)\in \Lambda_{q}(c_f)$ and $q\geq0$, the Taylor expansion of $\delta_{i,h}$ about location $s$ when $q\geq 1$ is:
{\allowdisplaybreaks 
\begin{align}
\label{eq:del-expan}
\delta_{i,h} & = \sum_{j=1}^{\lfloor q \rfloor}\frac{\mu^{(j)}_{s}}{j!} \left\{ (s_{i}-s)^{j}- (s_{i+h}-s)^{j} \right\}+ O\left(|s_{i}-s|^{q}+|s_{i+h}-s|^{q} \right) \nonumber \\
& = -\frac{h}{n}\sum_{j=1}^{\lfloor q \rfloor}\frac{\mu^{(j)}_{s}}{j!} \sum_{a=0}^{j-1}(s_{i}-s)^{a} (s_{i+h}-s)^{j-1-a} + O(|s_{i}-s|^{q} +|s_{i+h}-s|^{q}); 
\end{align}
and when $0 \leq q < 1$, it is:
\begin{equation}
\label{eq:del-expanq}
\delta_{i,h} = c\left(\frac{i}{n} \right)^{q} - c\left(\frac{i+h}{n} \right)^{q} = O\left(n^{-q}\right).
\end{equation}
}

A Taylor expansion of $g_{i,h}$ about location $s$ is:
{\allowdisplaybreaks 
\begin{align}
\label{eq:g-expan}
\frac{1}{2}g_{i,h}  = & (1-\rho_{h})\left[\sigma_{s}^{2}+\sigma_{s}\sum_{j=1}^{\lfloor\beta\rfloor}\frac{\sigma^{(j)}_{s}}{j!} \left\{\left(s_{i}-s\right)^{j}+\left(s_{i+h}-s\right)^{j}\right\} \right] +O(|s_{i}-s|^{\beta})\nonumber \\
& + \sum_{l=1}^{\lfloor \beta/2 \rfloor}\left( \frac{\sigma^{(l)}_{s}}{l!} \right)^{2}\left\{\left(s_{i}-s\right)^{2l}+\left(s_{i+h}-s\right)^{2l}-\rho_{h} \left(s_{i}-s\right)^{l}\left(s_{i+h}-s\right)^{l}\right\} \nonumber \\ 
& +  \sum_{m=3}^{\lfloor\beta\rfloor} \sum_{j=1}^{m-1} \left[\frac{c_{m}\sigma^{(j)}_{s}\sigma^{(m-k)}_{s}}{m!} \left\{\left(s_{i}-s\right)^{m}+\left(s_{i+h}-s\right)^{m}\right\} -\rho_{h} \frac{\sigma^{(j)}_{s}\sigma^{(m-j)}_{s}}{j! (m-j)!} \left(s_{i}-s\right)^{j}\left(s_{i+h}-s\right)^{m-j}\right]  
\end{align}}
under the condition that $\sigma^{2}(\cdot)\in \Lambda_{\beta}^{+}$ and $\beta \geq 2$.
Note that 
\begin{align}
\sum_{i=1}^{n-h}K_{\lambda,i}(s)&  = \sum_{i=1}^{n-h}\int_{(s_i+s_{i-1})/2}^{(s_i+s_{i+1})/2}\frac{1}{\lambda}K_{(B)}\left(\frac{s-u}{\lambda} \right)du \\ 
& = O(n\lambda)O\left(\frac{1}{n\lambda}\right)=O(1).
\label{eq:kernelO2}
\end{align}
Then
\begin{equation}
\label{eq:kernelO3} 
\sum_{i=1}^{n-h}K_{\lambda,i+\frac{h}{2}}^{2}(s)=O\left(\frac{1}{n\lambda}\right).
\end{equation}

\bibliographystyle{apalike}
\bibliography{thesis0}


\end{document}